\documentclass{article}

\usepackage[utf8]{inputenc}
\usepackage{mathtools,geometry,enumitem,bbm,changepage,tocloft}

\usepackage[bottom]{footmisc} 

\usepackage{amsmath,amsfonts,bm, amsthm}
\usepackage{amssymb}

\newcommand{\mc}[1]{\mathcal{#1}}
\newcommand{\bb}[1]{\mathbb{#1}}

\def\eqref#1{equation~\ref{#1}}

\def\1{\bm{1}}

\DeclareMathAlphabet{\mathsfit}{\encodingdefault}{\sfdefault}{m}{sl}
\SetMathAlphabet{\mathsfit}{bold}{\encodingdefault}{\sfdefault}{bx}{n}

\def\gA{{\mathcal{A}}}

\def\gD{{\mathcal{D}}}

\def\gU{{\mathcal{U}}}

\newcommand{\E}{\mathbb{E}}

\newcommand{\tr}{\ensuremath{\mathrm{tr}}}
\newcommand{\dt}{\ensuremath{\mathsf{D}}}
\newcommand{\qdt}{\ensuremath{\mathsf{Q}}}
\newcommand{\rdt}{\ensuremath{\mathsf{R}}} 
\newcommand{\bs}{\ensuremath{\mathsf{bs}}} 
\newcommand{\fbs}{\ensuremath{\mathsf{fbs}}} 

\newcommand{\corr}{\ensuremath{\mathsf{Corr}}}
\newcommand{\bicorr}{\ensuremath{\mathsf{biCorr}}}
\newcommand{\sel}{\ensuremath{\mathsf{Sel}}}

\newcommand{\xor}{\text{\scshape Xor}}
\newcommand{\myAnd}{\text{\scshape And}}
\newcommand{\lr}{\ensuremath{\mathsf{LR}}} 
\newcommand{\olr}{\ensuremath{\mathsf{OLR}}} 
\newcommand{\tllr}{\ensuremath{\mathsf{TLLR}}} 
\newcommand{\otllr}{\ensuremath{\mathsf{OTLLR}}} 
\newcommand{\settled}{\ensuremath{\mathsf{S}}} 
\newcommand{\progress}{\ensuremath{\mathsf{P}}} 
\newcommand{\depth}{\ensuremath{\mathrm{depth}}} 
\newcommand{\bootstrap}{\ensuremath{\mathrm{bootstrap}}} 
\newcommand{\Input}{\ensuremath{\mathsf{Input}}} 

\newtheorem{theorem}{Theorem}
\newtheorem{lemma}{Lemma}
\newtheorem{corollary}{Corollary}
\newtheorem{conjecture}{Conjecture}
\theoremstyle{definition}
\newtheorem{definition}{Definition}
\theoremstyle{remark}
\newtheorem{claim}{Claim}
\newtheorem{fact}{Fact}

\usepackage[greek,english]{babel}
\usepackage[dvipsnames]{xcolor}
\usepackage{hyperref}
\hypersetup{
	colorlinks=true,
	linkcolor=Sepia,
	citecolor=Sepia,
	filecolor=Sepia,
	urlcolor=Sepia
}
\addto\extrasenglish{}
\addto\extrasenglish{}
\addto\extrasenglish{}
\addto\extrasenglish{}
\addto\extrasenglish{}
\addto\extrasenglish{}
\addto\extrasenglish{}

\begin{document}

\newgeometry{margin=1.25in,top=1.7in,bottom=1in}

\begin{center}
{\LARGE The Power of Many Samples in Query Complexity}
\\[10mm]

\large
\setlength\tabcolsep{0em}
\newcommand{\myPad}{\hspace{3em}}
\begin{tabular}{c@{\myPad}c@{\myPad}c}
Andrew Bassilakis &
Andrew Drucker &
Mika G\"o\"os\\[-1mm]
\small\slshape Stanford University&
\small\slshape University of Chicago&
\small\slshape Stanford University\\[2mm]
Lunjia Hu &
Weiyun Ma &
Li-Yang Tan\\[-1mm]
\small\slshape Stanford University&
\small\slshape Stanford University&
\small\slshape Stanford University\\
\end{tabular}

\vspace{8mm}

{\large\itshape\today}

\vspace{8mm}

\normalsize
\bf Abstract
\end{center}

\noindent
The randomized query complexity $\rdt(f)$ of a boolean function $f\colon\{0,1\}^n\to\{0,1\}$ is famously characterized (via Yao's minimax) by the least number of queries needed to distinguish a distribution~$\gD_0$ over $0$-inputs from a distribution $\gD_1$ over $1$-inputs, maximized over all pairs~$(\gD_0,\gD_1)$. We ask: Does this task become easier if we allow query access to infinitely many samples from either $\gD_0$ or~$\gD_1$? We show the answer is \emph{no}: There exists a hard pair $(\gD_0,\gD_1)$ such that distinguishing $\gD_0^\infty$ from~$\gD_1^\infty$ requires $\Theta(\rdt(f))$ many queries. As an application, we show that for any composed function~$f\circ g$ we have $\rdt(f\circ g) \geq \Omega(\fbs(f)\rdt(g))$ where $\fbs$ denotes fractional block sensitivity.

\vspace{2mm}

\setlength{\cftbeforesecskip}{.7em}
\tableofcontents

\thispagestyle{empty}
\setcounter{page}{0}

\newpage
\restoregeometry

\section{Introduction}

Randomized query complexity (see~\cite{Buhrman2002} for a classic survey) is often studied using Yao's minimax principle~\cite{Yao1977}. The principle states that for every boolean function $f\colon\{0,1\}^n\to\{0,1\}$
\begin{center}
{\bf Yao's minimax:}
$\quad\rdt_\epsilon(f) ~=~ \max_\gD \, \dt_\epsilon(f,\gD)$.
\end{center}

\begin{itemize} \setlength{\itemsep}{0.5em}
	\item Here $\rdt_\epsilon(f)$ is the randomized $\epsilon$-error query complexity of $f$. More precisely, $\rdt_\epsilon(f)$ equals the least number of queries a randomized algorithm (decision tree) must make to the input bits $x_i\in\{0,1\}$ of an unknown input $x\in\{0,1\}^n$ in order to output $f(x)$ with probability at least $1-\epsilon$ (where the probability is over the coin tosses of the algorithm). We often set $\epsilon = 1/3$ and omit $\epsilon$ from notation, as it is well known that this choice only affects constant factors in query complexity.
	\item $\gD$ is a distribution over the inputs $\{0,1\}^n$. We may assume wlog that $\gD$ is \emph{balanced}: $\gD=\frac{1}{2}\gD_0+\frac{1}{2}\gD_1$ where $\gD_b$ is a distribution over $f^{-1}(b)$.
	\item $\dt_\epsilon(f,\gD)$ is the \emph{distributional} $\epsilon$-error query complexity of $f$ relative to $\gD$. More precisely, $\dt_\epsilon(f,\gD)$ equals the least number of queries a \emph{deterministic} algorithm must make to an input $x\sim \gD$ in order to output $f(x)$ with probability at least $1-\epsilon$ (where the probability is over $x\sim \gD$).
\end{itemize}

\subsection{Correlated samples problem} \label{sec:corr-samples}

One way to think about the distributional complexity of $f$ relative to $\gD=\frac{1}{2}\gD_0+\frac{1}{2}\gD_1$ is as the following task: A deterministic algorithm is given query access to a sample from either $\gD_0$ or $\gD_1$ and it needs to decide which is the case. In this work, we ask: Does this task become easier if we allow query access to an unlimited number of independent samples from either $\gD_0$ or $\gD_1$? In short,
\begin{center}\itshape
Is it easier to distinguish $\gD_0^\infty$ from $\gD_1^\infty$ than it is to distinguish $\gD_0$ from $\gD_1$?
\end{center}
More formally, we define the \emph{correlated samples problem} for $f$ relative to $\gD=\frac{1}{2}\gD_0+\frac{1}{2}\gD_1$ by
\[
\corr_\epsilon(f,\gD) ~\coloneqq~ \min_{k\geq 1}
\textstyle \, \dt_\epsilon(f^k,\frac{1}{2}\gD_0^k+\frac{1}{2}\gD_1^k).
\]
Here $f^k\colon(\{0,1\}^n)^k\to\{0,1\}^k$ is the function that evaluates $k$ copies of $f$ on disjoint inputs. We also use the notation $\gD^k\coloneqq\gD\times\cdots\times \gD$ ($k$ times) for the $k$-fold product distribution. In particular, under $\frac{1}{2}\gD_0^k+\frac{1}{2}\gD_1^k$, the function $f^k$ outputs either $0^k$ or $1^k$; the correlated samples problem is to decide which is the case. We note that the expression to be minimized on the right side is a non-increasing function of $k$ (access to more samples is only going to help). We may also assume wlog that $k\leq n$ (when an algorithm queries a sample for the first time, we may assume it is the first unqueried sample so far).

\paragraph{Shaltiel examples.}
It is not hard to give examples of input distributions where access to multiple correlated samples \emph{does} help. Such examples were already discussed by Shaltiel~\cite{Shaltiel2004} in the context of direct product theorems. For instance, consider the $n$-bit $\xor_n$ function. It is well known that $\rdt_\epsilon(\xor_n)=n$ for all~$\epsilon>0$. Define a balanced input distribution (here $\gU$ is a uniform random bit in $\{0,1\}$)
\[
\gD~\coloneqq~
\begin{cases}
0\,\gU^{n-1} & \text{with probability 99\%},\\
1\,\gU\,0^{n-2}  & \text{with probability 1\%}.
\end{cases}
\]
This distribution is hard 99\% of the time: if the first bit is 0, an algorithm has to compute $\xor_{n-1}$ relative to $\gU^{n-1}$, which requires $n-1$ queries. For the remaining 1\%, the distribution is easy: if the first bit is~1, the output can be deduced from the second bit. Here multiple correlated samples help a lot (for $\epsilon=1/3$):
\begin{align*}
\dt(\xor_n,\gD)~&=~\Omega(n),\\
\corr(\xor_n,\gD)~&=~O(1).
\end{align*}
Indeed, given a single sample from $\gD$, an algorithm is likely to have to solve the hard case of the distribution. By contrast, given multiple correlated samples, we can query the first bit for a large constant number of samples. This will give us a high chance to encounter at least one easy sample.

\paragraph{Error reduction.}
An important fact (which fails in the single-sample setting!)~is that we can amplify the success probability of any algorithm for correlated samples. This is achieved by a variant of the usual trick: repeatedly run the algorithm on fresh samples to gain more confidence about the output.\footnote{
In more detail: An algorithm $T$ with error $1/2-\delta$ has $|p_0-p_1|\geq2\delta$ where $p_i \coloneqq \Pr[T(x_i)=1]$ for $x_i\sim\gD_i^k$. Reducing error below~$\epsilon>0$ boils down to distinguishing two random coins with heads-probabilities $p_0$ and $p_1$. Given multiple samples from one of the coins, Chernoff bounds state that $O(\log(1/\epsilon)/\delta^2)$ samples are enough to tell which coin the samples came from.}
\begin{fact} \label{fact:amplify}
$\corr_\epsilon(f,\gD)\leq O(\log(1/\epsilon)/\delta^2)\cdot\corr_{1/2-\delta}(f,\gD)$ for every $(f,\gD)$. \qed
\end{fact}
The aforementioned Shaltiel example $(\xor_n,\gD)$ can alternatively be computed as follows: By querying the first two bits of a single sample $x\sim\gD$ one can predict $\xor_n(x)$ to within error $49.5\%$. Now apply \autoref{fact:amplify} to reduce the error below $1/3$ at the cost of a constant-factor blowup in query cost.

\subsection{Main result}
We study whether Shaltiel examples can be avoided if we restrict our attention to the hardest possible input distribution. Namely, we define a distribution-free complexity measure by
\[
\corr_\epsilon(f) ~\coloneqq~ \max_\gD \, \corr_\epsilon(f,\gD).
\]
Our main result is that multiple correlated samples do not help for the hardest distribution.

\begin{theorem} \label{thm:main}
$\corr(f) = \Theta(\rdt(f))$ for any (partial) boolean function $f$.
\end{theorem}

The main challenge in proving \autoref{thm:main} is precisely the existence of Shaltiel examples: How to construct hard distributions that do not contain any hidden easy parts? 
We resolve it by building decision trees that can exploit the easy parts not only in its own input distribution, but in various other distributions as well.

\subsection{Application 1: Selection problem}

Next we describe a consequence of our main result to a natural query task that we dub the \emph{selection problem}. A similar problem, called \emph{choose}, was studied by \cite{Beimel2014} in communication complexity.

\bigskip\noindent
Fix an $n$-bit function $f$ together with an input distribution $\gD$. In the \emph{$k$-selection problem} for $(f,\gD)$ the input is a random $kn$-bit string $x=(x^1,\ldots,x^k)\sim \gD^k$, and the goal is to output $(i,f(x^i))$ for some $i\in [k]$. That is, the algorithm gets access to $k$ independent samples from $\gD$ and it selects one of them to solve. We define
\begin{align*}
k\text{-}\sel_\epsilon(f,\gD) ~&\coloneqq~ \text{$\epsilon$-error query complexity of $k$-selection for $(f,\gD)$},\\
\sel_\epsilon(f,\gD) ~&\coloneqq~\textstyle \min_{k\geq 1} \, k\text{-}\sel_\epsilon(f,\gD), \\
\sel_\epsilon(f) ~&\coloneqq~\textstyle \max_{\gD} \, \sel_\epsilon(f,\gD).
\end{align*}
The selection problem is interesting because it, too, is subject to Shaltiel examples: for $(\xor_n,\gD)$ as described in \autoref{sec:corr-samples}, we have $\sel(\xor_n,\gD)=O(1)$ using the same idea of searching for an easy sample.

The following relates selection to correlated samples; see \autoref{sec:selection} for the proof.
\begin{theorem} \label{thm:selection}
The correlated samples problem is easier than selection:

\begin{enumerate}
	\item $\corr(f,\gD) \leq O(\sel(f,\gD))$ for every $(f,\gD)$.
	\item There exists an $n$-bit $(f,\gD)$ such that $\sel(f,\gD) = \Omega(n)$ but $\corr(f,\gD)=O(1)$.
	\item Selection does not admit efficient error reduction (as in~\autoref{fact:amplify}).
\end{enumerate}
\end{theorem}

Combining the first item of \autoref{thm:selection} with our main result (\autoref{thm:main}) we conclude that multiple samples do not help in the selection problem for the hardest distribution.
\begin{corollary}
$\sel(f) = \Theta(\rdt(f))$ for any (partial) boolean function $f$. \qed
\end{corollary}

\subsection{Application 2: Randomized composition}

We give another application of our main result to the \emph{randomized composition conjecture} studied in~\cite{Ben-David2016,Anshu2018,Gavinsky2019,Ben-David2020}. In fact, this application is what originally motivated our research project!

\bigskip\noindent
For an $n$-bit function $f$ and an $m$-bit function $g$ we define their composition
\[
f\circ g\colon (\{0,1\}^m)^n \to\{0,1\}
\qquad\text{such that}\qquad
(f\circ g)(x^1,\ldots,x^n)~\coloneqq~f(g(x^1),\ldots,g(x^n)).
\]
A \emph{composition theorem} aims to understand the query complexity of $f\circ g$ in terms of $f$ and $g$. Such theorems are known for deterministic query complexity, $\dt(f\circ g)=\dt(f)\dt(g)$~\cite{Savicky2002,Tal2013,Montanaro2014}, and quantum query complexity, $\qdt(f\circ g)=\Theta(\qdt(f)\qdt(g))$~\cite{Hoyer2007,Reichardt2011}. The conjecture in the randomized case is:
\begin{conjecture}
$\rdt(f\circ g) \geq \Omega(\rdt(f)\rdt(g))$ for all boolean functions $f$ and $g$.
\end{conjecture}
Gavinsky et al.~\cite{Gavinsky2019} have shown that the conjecture fails if $f$ is allowed to be a \emph{relation}. They also show $\rdt(f\circ g)\geq \Omega(\rdt(f)\rdt(g)^{1/2})$ for any relation $f$ and partial function~$g$. In a very recent work (concurrent to ours) Ben-David and Blais~\cite{Ben-David2020} have found a counterexample to the randomized conjecture for partial~$f$ and $g$, albeit with a tiny query complexity compared to input length; see also \autoref{sec:ind-work}. The conjecture is still open for total functions.

\paragraph{Fractional block sensitivity.}
We show a new composition theorem in terms of \emph{fractional block sensitivity}~$\fbs(f)$, introduced by~\cite{Tal2013,Gilmer2016}; see also~\cite{Kulkarni2016,Ambainis2018}. This measure is at most randomized query complexity, $\fbs(f)\leq O(\rdt(f))$, and it is equivalent to \emph{randomized certificate complexity}~\cite{Aaronson2008}.

Let us define $\fbs(f)$ for an $n$-bit $f$. We say that a block $B\subseteq[n]$ is \emph{sensitive} on input $x$ iff $f(x)\neq f(x^B)$ where $x^B$ is~$x$ but with bits in $B$ flipped. Fix an input $x$ and introduce a real weight $w_B\in[0,1]$ for each sensitive block $B$ of $x$. Define $\fbs(f,x)$ as the optimum value of the following linear program
\begin{center}
\begin{tabular}{rll}
$\max$ & $\sum_B w_B$ \\[2mm]
$\textit{subject to}$
& $\sum_{B\ni i} w_B \leq 1$, & $\forall i\in[n]$, \\[1mm]
& $w_B\geq 0$, & $\forall B$.
\end{tabular}
\end{center}
Finally, define $\fbs(f)\coloneqq \max_x \fbs(f,x)$. For comparison, the more usual \emph{block sensitivity} $\bs(f)$~\cite{Nisan1991} is defined the same way except with the integral constraint $w_B\in\{0,1\}$. In particular $\bs(f)\leq \fbs(f)$, and moreover a polynomial gap (power $1.5$) between the two is known for a total function~\cite{Gilmer2016}.

\bigskip\noindent
We make progress towards the composition conjecture; see~\autoref{sec:composition} for the proof.
\begin{theorem} \label{thm:composition}
$\rdt(f\circ g)\geq\Omega(\fbs(f) \rdt(g))$ for any (partial) boolean functions $f$ and $g$.
\end{theorem}

The previous best comparable composition theorem was $\rdt(f\circ g)\geq\Omega(\bs(f)\rdt(g))$, a proof of which is virtually the same as for the result that $\rdt(\myAnd_n\circ g)\geq\Omega(n\rdt(g))$; see~\cite[\S5.1]{Goos2018}. In fact, we were originally motivated to consider the correlated samples problem when trying to strengthen this composition result from block sensitivity to fractional block sensitivity.

\subsection{Independent work by Ben-David and Blais} \label{sec:ind-work}

In an independent and concurrent work, Ben-David and Blais~\cite{Ben-David2020} have also studied the randomized composition conjecture and ways of circumventing Shaltiel examples via improved minimax theorems. They develop a powerful framework for constructing hard \emph{Shaltiel-free} distributions, which is general enough to apply not only to query complexity but also, for instance, to communication complexity. In particular, their framework is able to give an alternative proof of our main result (\autoref{thm:main}) as well as our $\fbs$-based composition theorem (\autoref{thm:composition}). Their proof techniques involve information theory and analysis; by contrast, our techniques are more elementary and directly tailored to the correlated samples problem (which does not explicitly appear in their work).
\subsection{Roadmap}
\label{sec:roadmap}
We will prove our main theorem (\autoref{thm:main}) in \autoref{sec:booster} and \autoref{sec:bootstrapping}. Before that, we introduce our basic notions regarding decision trees in \autoref{sec:preliminaries}. In \autoref{sec:booster}, we characterize decision trees as \emph{likelihood boosters}, emphasizing that a good query algorithm must make significant progress in terms of boosting the likelihood of one of the outputs (0 or 1) to much higher than the other, and vice versa. This characterization frees us from considering inputs from both $\gD_0$ and $\gD_1$ simultaneously: if an algorithm is certain about the output on $\gD_1$, then it must also make few errors on $\gD_0$. We thus reduce the proof of \autoref{thm:main} to bootstrapping decision trees that can make \emph{overall} progress across multiple samples to a decision tree that makes uniform progress. In \autoref{sec:bootstrapping}, we build such a bootstrapping algorithm and show that it makes satisfactory progress with a careful analysis. Proofs for our two applications are in \autoref{sec:selection} and \autoref{sec:composition}.
\section{Preliminaries}
\label{sec:preliminaries}
Let $f: \Sigma^n \to \{0,1,*\}$ be a partial function for some alphabet $\Sigma$ (typically $\Sigma=\{0,1\}$). Let $\gD_0$, $\gD_1$ be distributions supported on $f^{-1}(0)$, $f^{-1}(1)$ respectively. For each $x\in\Sigma^n$, let $\gD_0(x)$ (resp.\ $\gD_1(x)$) denote the probability mass on $x$ in distribution $\gD_0$ (resp.\ $\gD_1$). For a subset $S\subseteq \Sigma^n$, we define $\gD_b(S) = \sum_{x\in S}\gD_b(x)$ for $b = 0,1$. If $\gD_b(S)>0$, we define the conditional distribution $\gD_b|_S$ by $\gD_b|_S(x) = \frac{\gD_b(x)}{\gD_b(S)}$ when $x\in S$, and $\gD_b|_S(x) = 0$ when $x\notin S$. We define the \emph{likelihood-ratio} of $S$ as $$\lr(S) := \frac{\gD_1(S)}{\gD_0(S)}.$$

Let $T$ be a deterministic decision tree that takes as input a sample $x \in \Sigma^n$ drawn from either $\gD_0$ or $\gD_1$. For every vertex $v$ in $T$, we use $\Input(v)\subseteq \Sigma^n$ to denote the set of strings that can reach $v$, or equivalently, the set of strings that agree with all the queries made so far. Typically, every non-leaf vertex in $T$ corresponds to a query to a certain position in the sample, but we will allow non-leaf vertices $v$ in $T$ that do not make any query, each of them having only a single child $v'$ with $\Input(v') = \Input(v)$. We abuse notation slightly and use $v$ as a shorthand for $\Input(v)$, so we have $\gD_0(v) = \sum_{x\in v}\gD_0(x)$, $\gD_1(v) = \sum_{x\in v}\gD_1(x)$ and $$\lr(v) = \frac{\gD_1(v)}{\gD_0(v)}.$$ Note that the likelihood-ratio $\lr(v)$ is non-negative, but could be zero or infinite. We can eliminate the undefined case ($\gD_0(v) = \gD_1(v)=0$) by trimming the unreachable parts of the decision tree.

Now if the decision tree $T$ takes as input $k$ samples from $\Sigma^n$, it is not hard to see that $\Input(v)$ can be written as a Cartesian product $\Input(v) = \Input_1(v)\times \cdots\times\Input_k(v)$, where $\Input_j(v)\subseteq\Sigma^n$ is the set of strings that agree with all the queries made to the $j$-th sample so far. Again, we abuse notation slightly and use $v_j$ as a shorthand for $\Input_j(v)$, so we will often write $v = v_1\times\cdots\times v_k$. We define the \emph{overall likelihood ratio} of $v$ as the product $$\olr(v) := \lr(v_1)\cdots\lr(v_k) = \frac{\gD_1(v_1)}{\gD_0(v_1)}\cdots\frac{\gD_1(v_k)}{\gD_0(v_k)}.$$

It is often more convenient to consider the logarithm of likelihood ratios. We will use natural logarithm throughout the paper, i.e. $\log(\cdot) = \ln(\cdot)$.
\section{Query Algorithms as Likelihood Boosters}
\label{sec:booster}
Our overarching goal (\autoref{thm:main}) is to construct an efficient deterministic query algorithm that distinguishes $\gD_0$ from $\gD_1$, assuming the existence of one that distinguishes $\gD_0^k$ from $\gD_1^k$. As the starting point, we introduce the notion of \emph{likelihood boosters} as a way of measuring the progress made by a query algorithm $T$ in distinguishing $\gD_0$ from $\gD_1$. The key idea is that, as more queries are being made, the algorithm narrows down the possibilities of the unknown input, driving the likelihood of one of the output (0 or 1) much higher than the other. In fact, we show that $T$ can distinguish $\gD_0$ from $\gD_1$ well if and only if a sample drawn from $\gD_1$ has a high probability of arriving at a leaf of $T$ where most of the remaining possibilities produce output 1. (\autoref{lm:booster->query} and \autoref{lm:query->booster}). 

In the multiple-sample setting, we use the notions of \emph{overall likelihood boosters} and \emph{uniform likelihood boosters}, which have different levels of guarantees, to measure the progress of a query algorithm on simultaneously classifying each of the samples in the input. We show that an efficient query algorithm that distinguishes $\gD_0^k$ from $\gD_1^k$ is an efficient overall likelihood booster (\autoref{cor:corr->overall}). Moreover, we show that an efficient uniform likelihood booster on multiple samples induces an efficient likelihood booster on a single sample (\autoref{lm:uniform->single}), which in turn implies an efficient query algorithm that distinguishes $\gD_0$ from $\gD_1$ (\autoref{lm:booster->query}). These results will enable us to reduce proving \autoref{thm:main} to relating overall likelihood boosters to uniform likelihood boosters, which is the focus of \autoref{sec:bootstrapping} (see \autoref{thm:bootstrapping}). 

We now formally define the three types of likelihood boosters mentioned above: 



\begin{definition}
\label{def:booster}
We say a deterministic decision tree $T$ is a \emph{($\delta, M$)-likelihood booster} for $\gD_0,\gD_1$ if, with probability at least $1-\delta$, an input sample drawn from $\gD_1$ reaches a leaf $\ell$ of $T$ with likelihood ratio $\lr(\ell)\geq M$.
\end{definition}

\begin{definition}
\label{def:overall}
We say a deterministic decision tree $T$ is a \emph{($\delta, M$)-overall likelihood booster} for $\gD_0^k,\gD_1^k$ if, with probability at least $1-\delta$, an input drawn from $\gD_1^k$ consisting of $k$ samples reaches a leaf $\ell$ of $T$ with \emph{overall} likelihood ratio $\olr(\ell)\geq M$.
\end{definition}

\begin{definition}
\label{def:uniform}
We say a deterministic decision tree $T$ is a $(\delta,\varepsilon,M)$-\emph{uniform likelihood booster} for $\gD_0^k$ and $\gD_1^k$ if, with probability at least $1-\delta$, an input $x$ drawn from $\gD_1^k$ consisting of $k$ samples reaches a leaf $\ell = \ell_1\times\cdots\times\ell_k$ of $T$ with the property that at least $(1-\varepsilon)k$ different samples $j\in\{1,\cdots,k\}$ satisfy $\lr(\ell_j)\geq M$.
\end{definition}

Note that the above definitions do not depend on the actual output of the decision tree $T$. We now show in the following two lemmas that likelihood boosters are in some sense equivalent to query algorithms that distinguish $\gD_0$ from $\gD_1$. 


\begin{lemma}
\label{lm:booster->query}
Suppose $T$ is a $(\delta, M)$-likelihood booster for $\gD_0,\gD_1$. Consider the deterministic decision tree $T'$ that makes exactly the same queries as $T$ and accepts if and only if a leaf $\ell$ with $\lr(\ell)\geq M$ is reached. Then $T'$ distinguishes $\gD_0$ from $\gD_1$ with the following guarantees:
\begin{enumerate}
    \item (Completeness) $T'$ accepts $x\sim\gD_1$ with probability at \emph{least} $1-\delta$.
    \item (Soundness) $T'$ accepts $x\sim\gD_0$ with probability at \emph{most} $1/M$.
\end{enumerate}
\end{lemma}

\begin{proof}
Completeness follows directly from the definition of likelihood booster. To prove soundness, consider the set $U$ of leaves $\ell$ with $\lr(\ell)\geq M$. For all $\ell\in U$, we have $\gD_0(\ell)\leq \frac 1M\gD_1(\ell)$. Therefore, $\sum_{\ell \in U}\gD_0(\ell)\leq \frac 1M\sum_{\ell \in U}\gD_1(\ell)\leq \frac 1M$. This means that a sample from $\gD_0$ reaches leaves in $U$ with probability at most $\frac 1M$, which is exactly the desired soundness.
\end{proof}

\begin{lemma}
\label{lm:query->booster}
Suppose a deterministic decision tree $T$ can distinguish $\gD_0$ from $\gD_1$ with the following guarantees: $T$ accepts $x\sim\gD_0$ with probability at most $\delta_0$, and accepts $x\sim\gD_1$ with probability at least $1-\delta_1$. Then $T$ is a $(M\delta_0 + \delta_1, M)$-likelihood booster for any $M>0$.
\end{lemma}
\begin{proof}
Let $U$ denote the set of leaves $\ell$ with $\lr(\ell)< M$. We can partition $U$ as $U = U_0 \cup U_1$, where $U_1$ corresponds to the leaves at which $T$ accepts. Since $T$ accepts with probability at most $\delta_0$ on $\gD_0$, we have $\sum_{\ell\in U_1}\gD_0(\ell)\leq \delta_0$. Similarly, we have $\sum_{\ell\in U_0}\gD_1(\ell)\leq \delta_1$. Therefore,
\begin{equation*}
    \sum_{\ell\in U}\gD_1(\ell) = \sum_{\ell\in U_0}\gD_1(\ell) + \sum_{\ell\in U_1}\gD_1(\ell)\leq \sum_{\ell\in U_0}\gD_1(\ell) + M\sum_{\ell\in U_1}\gD_0(\ell) = \delta_1 + M\delta_0.
\end{equation*}
In other words, a sample from $\gD_1$ has probability at most $M\delta_0 + \delta_1$ of reaching a leaf in $U$, which means that $T$ is a $(M\delta_0 + \delta_1, M)$-likelihood booster.
\end{proof}

In the multiple-sample setting, if we view the pair $\gD_0^k$ and $\gD_1^k$ as $\gD_0'$ and $\gD_1'$ in the single-sample setting with input length multiplied by $k$, the definition of overall likelihood ratio coincides with the definition of likelihood ratio in the single-sample setting. Therefore, we have the following corollary of \autoref{lm:query->booster}, which essentially shows that an efficient query algorithm for the correlated samples problem is an efficient overall likelihood booster:

\begin{corollary}
\label{cor:corr->overall}
Suppose a deterministic decision tree $T$ can distinguish $\gD_0^k$ from $\gD_1^k$ in that $T$ accepts $x\sim\gD_0^k$ with probability at most $\delta_0$, and $T$ accepts $x\sim\gD_1^k$ with probability at least $1-\delta_1$. Then $T$ is a $(M\delta_0 + \delta_1, M)$-overall likelihood booster for any $M>0$.
\end{corollary}

To conclude this section, we show that an efficient uniform likelihood booster in the multiple-sample setting implies an efficient likelihood booster in the single-sample setting. 


\begin{lemma}
\label{lm:uniform->single}
For any $(\delta,\varepsilon,M)$-\emph{uniform likelihood booster} $T$ for $\gD_0^k$ and $\gD_1^k$ and any $C>0$, there is a $(\delta + \varepsilon + \frac 1C,M)$-likelihood booster $T'$ for $\gD_0$ and $\gD_1$ with $\depth(T')\leq C\cdot\frac{\depth(T)}k$.
\end{lemma}
\begin{proof}
Define $Q = C\cdot\frac{\depth(T)}k$. We first build a randomized query algorithm $\gA'$ for $\gD_0$ and $\gD_1$, and later derandomize it as $T'$. On input $x_{\gA'}$, $\gA'$ generates $k$ random samples $(x_1, \ldots, x_k) \sim \gD_1^k$, selects a uniformly random index $j$, replaces $x_j$ with $\gA'$'s own input $x_{\gA'}$, and finally simulates $T$ on the modified $k$ samples $(x_1, \ldots, x_{\gA'}, \ldots, x_k)$. If $T$ attempts to make the $(\lfloor Q\rfloor + 1)$-th query to the $j$-th (modified) sample, $\gA'$ halts. 

It is easy to see that the maximum number of queries made by $\gA'$ is at most $Q$. Moreover, by Markov's inequality, if the input $x_{\gA'}$ to $\gA'$ is drawn from $\gD_1$, the probability that $\gA'$ halts early because of $T$ making more than $Q$ queries to the $j$-th sample is at most $\frac 1C$, since the average number of queries $T$ makes to the $j$-th sample for a uniformly random $j$ is at most $\frac{\depth(T)}{k}$.

We now show that with probability at least $1-(\delta + \varepsilon + \frac 1C)$, $\gA'$ reaches a leaf $\ell = \ell_1\times\cdots\times\ell_k$ of $T$ with $\lr(\ell_j)\geq M$ when its own input $x_{\gA'}$ is drawn from $\gD_1$. By a union bound, we only need to show that this holds with probability at least $1-(\delta + \varepsilon)$ for the extended version of $\gA'$ that doesn't halt early. If we switch the order of randomness so that $j$ is chosen after a leaf of $T$ is reached, this follows easily from the definition of uniform likelihood boosters (\autoref{def:uniform}).

Finally, we derandomize $\gA'$. Note that the randomness in $\gA'$ only comes from the randomness in $j$ and in all the generated samples $x_{i}$ except the $j$-th sample. We can simply fix them so that the probability of reaching a leaf $\ell$ of $T$ with $\lr(\ell_j) \geq M$ is maximized, assuming that the $j$-th sample is from $\gD_1$. Since $j$ and all generated samples other than the $j$-th sample have been fixed, the decision tree $T$ now ``shrinks'' to a decision tree $T'$ with only the first $\lfloor Q\rfloor$ queries to the $j$-th sample remaining, and every leaf $\ell$ of $T$ that is reachable when we run $\gA'$ now becomes a leaf $\ell'$ of $T'$. Shrinking the tree doesn't affect the computation history regarding the $j$-th sample, so we have $\ell' = \Input(\ell') = \Input_j(\ell) = \ell_j$ and $\lr(\ell') = \lr(\ell_j)$. This proves that $T'$ is a ($\delta + \varepsilon + \frac 1C,M$)-likelihood booster.
\end{proof}

\allowdisplaybreaks
\section{Bootstrapping Overall Booster to Uniform Booster}
\label{sec:bootstrapping}

The results from the previous section (\autoref{sec:booster}) reduce proving our main result (\autoref{thm:main}) to proving relations between overall likelihood boosters and uniform likelihood boosters. In this section, we complete this step with the following result:


\begin{theorem}
\label{thm:bootstrapping}
Assume that there is a depth-$L$ $(0.1, 25)$-overall likelihood booster for \emph{every} distribution pair $\gD_0^k,\gD_1^k$. Then there is a depth-$O(KL)$ $(0.1, 0.1, 100)$-uniform likelihood booster for every distribution pair $\gD_0^K,\gD_1^K$ whenever $K \geq 1000 k(|\Sigma| + 1)^n$.
\end{theorem}

We first show how to derive \autoref{thm:main} from \autoref{thm:bootstrapping}:

\begin{proof}[Proof of \autoref{thm:main}]
We prove the inequality $\rdt(f) \leq O(\corr(f))$ (the converse inequality is trivial). Suppose we have a depth-$L$ deterministic decision tree that solves the correlated samples problem on $\frac 12\gD_0^k + \frac 12\gD_1^k$ with success probability at least $0.999$ (recall that the success probability can be amplified by \autoref{fact:amplify}). That is, the decision tree accepts inputs drawn from $\gD_1^k$ with probability at \emph{least} $0.998$ and accepts inputs drawn from $\gD_0^k$ with probability at \emph{most} $0.002$. By \autoref{cor:corr->overall}, it is a $(0.1, 25)$-overall likelihood booster for $\gD_0^k$ and $\gD_1^k$.

By \autoref{thm:bootstrapping}, for any pair of distributions $\gD_0^K, \gD_1^K$, there is a $(0.1, 0.1, 100)$-uniform likelihood booster with depth $O(KL)$. Then by \autoref{lm:uniform->single}, there is a $(1/3, 100)$-likelihood booster with depth $O(L)$ for $\gD_0$ and $\gD_1$, which by \autoref{lm:booster->query} implies a query algorithm for $\gD_0$ and $\gD_1$ with success probability at least $1/3$. By the arbitrariness of $\gD_0$ and $\gD_1$, we have $\rdt_{1/3}(f) = O(L)$ via Yao's minimax, as desired.
\end{proof}

The rest of this section is dedicated to proving \autoref{thm:bootstrapping}. We construct the desired uniform likelihood booster $T_\bootstrap$, described in \autoref{subsec:bootstrapping}, by applying different overall likelihood boosters to appropriate sets of samples at different phases of computation. To quantify the progress made by $T_\bootstrap$, we design a measure based on a “truncated” log likelihood ratio which handles samples that $T_\bootstrap$ is confident about with special care. As the technical core of the proof, we show that under our carefully constructed measure, $T_\bootstrap$ in expectation makes \emph{positive} and \emph{constant} progress during each phase of computation (Lemmas \ref{lm:submartingale} and \ref{lm:progress}). Therefore, $T_\bootstrap$ is able to achieve the desired guarantees after sufficiently many phases.   

\subsection{Bootstrapping algorithm}
\label{subsec:bootstrapping}
We describe our depth-$O(KL)$ $(0.1, 0.1, 100)$-uniform likelihood booster $T_\bootstrap$ taking $K \geq 1000k(|\Sigma| + 1)^n$ samples.
%
%
Recall that each vertex $v$ of $T_\bootstrap$ can be written as a Cartesian product $v = v_1\times \cdots\times v_K$, where $v_j\subseteq \Sigma^n$ is the set of strings that are consistent with the queries made to the $j$-th sample so far. We say that the $j$-th sample is \emph{settled} at $v$ if $$\lr(v_j) = \frac{\gD_1(v_j)}{\gD_0(v_j)}\notin[e^{-100}, e^{100}].$$ Note that it is possible for a sample to be settled in the wrong direction (e.g.\ $\lr(v_j)< e^{-100}$ on input drawn from $\gD_1^K$), but we will show that this is not a serious issue.

The query algorithm $T_\bootstrap$ proceeds in at most $C \cdot K$ phases (for some large constant $C>0$). Each phase consists of at most $L$ queries and is described as follows:\\\\
\textbf{Phase $s = 1,\cdots,C\cdot K$:}
\begin{enumerate}
\item If fewer than $k  (|\Sigma| + 1)^n$ out of the $K$ samples are unsettled, halt.
\item Else, since each $v_j$ is determined by a string $v_*$ in $(\Sigma\cup\{*\})^n$ recording the queries made so far to the $j$-th sample, by the Pigeonhole Principle there exist $k$ unsettled samples $j_1,\cdots,j_k$ with $v_{j_1} = \cdots = v_{j_k} = v_*$.

\item Run the depth-$L$ $(0.1, 25)$-overall likelihood booster $A^{(v_*)}$, assumed in \autoref{thm:bootstrapping} to exist,  relative to the input-distribution pair
\[   (\mc{D}_0|_{v_*})^k \ ,  \quad    (\mc{D}_1|_{v_*})^k \  \]
on the samples
\[  (x^{j_1}, \ldots, x^{j_k}) \ .   \]
If any query causes one of these samples to become settled (i.e. $\lr(v_{j_i})\notin[e^{-100}, e^{100}]$ for some $i\in\{1,\cdots,k\}$), halt $A^{(v_*)}$ and go to the next Phase. Otherwise we proceed to the next Phase after $A^{(v_*)}$ terminates. If fewer than $L$ queries are made in the current phase, insert dummy vertices that do not make any query (see \autoref{sec:preliminaries}) to $T_\bootstrap$ so that each phase corresponds to a  path in $T_\bootstrap$ with length exactly $L$.

\end{enumerate}
\subsection{Sub-martingale property of progress measure}
\label{subsec:submartingale}
It's not hard to see that the overall likelihood ratio ($\olr$) is \emph{not} an effective measure of progress for $T_\bootstrap$: $\olr$ can rocket to infinity even when there is only one settled sample. In this subsection, we introduce a better progress measure: \emph{overall truncated log likelihood ratio} ($\otllr$), and show that it is a sub-martingale along the computation path of \emph{any} decision tree (\autoref{lm:submartingale}). In other words, $T_{\bootstrap}$ always makes \emph{non-negative} progress in expectation. We will show that each phase of $T_\bootstrap$ makes \emph{positive} expected progress in the next subsection (\autoref{subsec:positive_progress}).

Let $T$ be a deterministic decision tree that takes as input $K$ samples. For every vertex $v = v_1\times\cdots\times v_K$ of $T$, we define the \emph{truncated log likelihood ratio} of $v_j$ as
\begin{equation*}
\tllr(v_j) := \left\{\begin{array}{ll} \log(\lr(v_j)),& \textup{if }|\log(\lr(v_j))|\leq 100, \\ 500, & \textup{otherwise.}\end{array}\right.
\end{equation*}
Note that if $\log(\lr(v_j))$ slightly exceeds the upper threshold 100, we set $\tllr$ to a much higher value 500. Also, when $\log(\lr(v_j))$ drops below the lower threshold -100, we also set $\tllr$ to 500. Thus, the $j$-th sample is \emph{settled} at $v$ if and only if $\tllr(v_j) = 500$.

We define the \emph{overall truncated log-likelihood-ratio} of $v$ as the sum
\[
\otllr(v) := \sum_{j=1}^K\tllr(v_j).
\]

The input $x$ to $T$ determines a computation path from the root of $T$ to a leaf: $v^0\rightarrow v^1\rightarrow \cdots\rightarrow v^q$. The randomness in $x$ transfers to the randomness in the path, so the path is a stochastic process. We now show that $\otllr(v^t)$ along the path is a sub-martingale when $x$ is drawn from $\gD_1^K$:
\begin{lemma}
\label{lm:submartingale}
Assume that $T$ never queries a settled sample. Assume that the input $x$ to $T$ is drawn from $\gD_1^K$, $v$ is a non-leaf vertex with distance $t$ from the root, and $v$ is reachable (i.e. $\Pr[v^t = v]>0$ on $\gD_1^K$). Define $\Delta^t := \otllr(v^{t+1}) - \otllr(v^t)$. Then we have
\begin{equation*}
\bb E[\Delta^t|v^t = v] \geq 0.001\cdot \bb E[(\Delta^t)^2|v^t = v]\geq 0.
\end{equation*}
\end{lemma}

\begin{proof}
Let us condition on $v^t = v$ in the whole proof. If $v$ is a dummy vertex that does not make any query, then $\Delta^t = 0$ deterministically and the lemma holds trivially. We assume that $v$ is not a dummy vertex henceforth.

Suppose sample $j$ is queried at vertex $v$. We have $\otllr(v^{t+1}) - \otllr(v^t) = \tllr(v^{t+1}_j) - \tllr(v^{t}_j)$. Since $T$ never queries a settled sample, we know $\tllr(v^t_j) = \log \frac{\gD_1(v^t_j)}{\gD_0(v^t_j)}\in [-100, 100]$.

Let $\sigma\in\Sigma$ denote the random outcome of the query, and let $p_0(\sigma), p_1(\sigma)$ denote the probability that the outcome to the query is $\sigma$ under $\gD_0|_{v_j^t}, \gD_1|_{v_j^t}$, respectively. Let $H\subseteq \Sigma$ denote the set of $\sigma\in \Sigma$ with $|\tllr(v^{t}_j) + \log\frac{p_1(\sigma)}{p_0(\sigma)}|>100$. Note that $\gD_0(v^{t+1}_j) = \gD_0(v^t_j)p_0(\sigma)$ and $\gD_1(v^{t+1}_j) = \gD_1(v^t_j)p_1(\sigma)$, so
\begin{equation*}
\tllr(v^{t+1}_j) = \left\{\begin{array}{ll}
\tllr(v^{t}_j) + \log\frac{p_1(\sigma)}{p_0(\sigma)}, & \sigma\notin H, \\
 500, & \sigma\in H.
\end{array}\right.
\end{equation*}
Thus, $H$ is precisely the set of outcomes $\sigma \in \Sigma$ that make sample $j$ settled at $v^{t+1}$. Let $W = W(\sigma)$ denote the difference $\tllr(v^{t+1}_j) - \tllr(v^{t}_j)$. Our goal is to prove $\bb E[W]\geq 0.001\cdot \bb E[W^2]$.

%
Note that $W(\sigma)\in [400, 600]$ when $\sigma \in H$ and $W(\sigma) = \log\frac{p_1(\sigma)}{p_0(\sigma)}\in [-200,200]$ when $\sigma\notin H$.
We have
\begin{align}
\bb E[W] \geq & 400 \sum_{\sigma\in H}p_1(\sigma) + \sum_{\sigma\notin H}p_1(\sigma)\log\frac{p_1(\sigma)}{p_0(\sigma)}\nonumber \\
= & 400 \sum_{\sigma\in H}p_1(\sigma) + \sum_{\sigma\notin H}p_0(\sigma)\cdot \frac{p_1(\sigma)}{p_0(\sigma)}\log\frac{p_1(\sigma)}{p_0(\sigma)}.\label{eq:submartingale}
\end{align}
By a helper lemma (\autoref{lm:helper}) proved in \autoref{sec:helper}, we know that $$\frac{p_1(\sigma)}{p_0(\sigma)}\log\frac{p_1(\sigma)}{p_0(\sigma)}\geq \left(\frac{p_1(\sigma)}{p_0(\sigma)} - 1\right) + \frac 1{400}\cdot \frac{p_1(\sigma)}{p_0(\sigma)}\left(\log\frac{p_1(\sigma)}{p_0(\sigma)}\right)^2.$$ Plugging this into (\ref{eq:submartingale}), we have
\begin{align*}
\bb E[W] \geq & 400 \sum_{\sigma\in H}p_1(\sigma) + \sum_{\sigma\notin H}p_1(\sigma) - \sum_{\sigma\notin H}p_0(\sigma) + \frac 1{400}\sum_{\sigma\notin H}p_1(\sigma)\left(\log\frac{p_1(\sigma)}{p_0(\sigma)}\right)^2\\
\geq & 400 \sum_{\sigma\in H}p_1(\sigma) +
\left(\sum_{\sigma\notin H}p_1(\sigma) - 1\right) + \frac 1{400}\sum_{\sigma\notin H}p_1(\sigma)\left(\log\frac{p_1(\sigma)}{p_0(\sigma)}\right)^2\\
= & 400 \sum_{\sigma\in H}p_1(\sigma) -
\sum_{\sigma\in H}p_1(\sigma) + \frac 1{400}\sum_{\sigma\notin H}p_1(\sigma)\left(\log\frac{p_1(\sigma)}{p_0(\sigma)}\right)^2\\
= & 399 \sum_{\sigma\in H}p_1(\sigma) + \frac 1{400}\sum_{\sigma\notin H}p_1(\sigma)\left(\log\frac{p_1(\sigma)}{p_0(\sigma)}\right)^2\\
= & 399 \sum_{\sigma\in H}p_1(\sigma) + \frac 1{400}\sum_{\sigma\notin H}p_1(\sigma)(W(\sigma))^2\\
\geq & \frac {1}{1000}\sum_{\sigma\in H}p_1(\sigma)(W(\sigma))^2+ \frac 1{400}\sum_{\sigma\notin H}p_1(\sigma)(W(\sigma))^2\\
\geq & \frac 1{1000}\bb E[W^2].
\end{align*}
\end{proof}

\subsection{Bounding the conditional expectation of progress}
\label{subsec:positive_progress}
In the previous subsection, we showed that $\otllr$, as a progress measure, is a sub-martingale. Now we refine our progress measure to also include the natural measure \emph{number of settled samples}, and show that each phase of $T_\bootstrap$ makes \emph{positive} progress in expectation.

Recall that we inserted dummy vertices in $T_\bootstrap$ to ensure that each phase corresponds to a computation path of length exactly $L$. Therefore, an entire computation path of $T_\bootstrap$ must have length divisible by $L$: $v^0\rightarrow \cdots \rightarrow v^{qL}$. The sub-path $v^{tL}\rightarrow \cdots\rightarrow v^{(t+1)L}$ is the computation path of the $(t+1)$-th phase.


Define $\settled(v)$ as the number of settled samples at vertex $v$. Our new measure of progress is $$\progress(v^t) := \settled(v^t) + \otllr(v^t).$$ 

\begin{lemma}
\label{lm:progress}
Assume that the input $x$ to $T_\bootstrap$ is drawn from $\gD_1^K$, $v$ is a non-leaf vertex with distance $tL$ from the root, and $v$ is reachable (i.e. $\Pr[v^{tL} = v]>0$ on $\gD_1^K$). Then we have
\[
\bb E[\progress(v^{(t+1)L}) - \progress(v^{tL})|v^{tL} = v] \geq 0.001.
\]
\end{lemma}

Before proving the lemma, we first show how it implies \autoref{thm:bootstrapping}.
\begin{proof}[Proof of \autoref{thm:bootstrapping}]
  We consider an extended version of $T_\bootstrap$ that always halts after exactly $C\cdot K$ phases: whenever it would halt at line 1, it instead enters dummy phases and increases its total progress $\progress$ by 0.001 per phase (so that now $\progress = \settled + \otllr + 0.001\cdot\textup{number of dummy phases}$). By \autoref{lm:progress}, the extended algorithm finishes with expected total progress $\bb E[\progress]\geq 0.001C\cdot K$ on input drawn from $\gD_1^K$. However, $\progress$ can never grow too large: before any dummy phase, $\progress$ is at most $501K$, and there are at most $C\cdot K$ dummy phases, so $\progress \leq 501K + 0.001C\cdot K$. By Markov's inequality on the non-negative random variable $(501K + 0.001C\cdot K) - \progress$, we have $\Pr[\progress\leq 501K]\leq \frac{501K}{0.001C\cdot K} = \frac{501}{0.001C}$. If we choose a large enough $C$, we know that with probability at least $0.99$, the total progress exceeds $501K$, which means that the extended algorithm enters dummy phases before halting, and the original algorithm halts at line 1 with all but 0.001 fraction of the samples settled.

It now suffices to show that the fraction of samples settled in the wrong direction is at most $0.01$ with probability at least $0.99$. We first fix $j$ and show that the probability that the $j$-th sample is settled in the wrong direction is at most $e^{-100}$, and then use the linearity of expectation and Markov's inequality to bound the overall wrong settlement.

Conditioning on all but the $j$-th sample, $T_\bootstrap$ becomes a deterministic decision tree $T'$ on a single sample. Let $U$ denote the set of leaves $\ell$ of $T'$ with $\lr(\ell)\leq e^{-100}$. We have $\sum_{\ell\in U}\gD_1(\ell)\leq e^{-100}\sum_{\ell\in U}\gD_0(\ell)\leq e^{-100}$. This means that the probability that a sample from $\gD_1$ reaches leaves in $U$ is at most $e^{-100}$. Thus the probability of wrong settlement for sample $j$ in $T_\bootstrap$ is at most $e^{-100}$.

By the linearity of expectation, the expected fraction of samples settled in the wrong direction is at most $e^{-100}$. Then by Markov's inequality, with probability at least $0.99$, the fraction of wrong settlement is at most $0.01$.
\end{proof}

\begin{proof}[Proof of \autoref{lm:progress}]
$\settled(v^{(t+1)L}) - \settled(v^{tL})$ is either $0$ or $1$, depending on whether or not a sample becomes settled in phase $t+1$. 

In the case where $\Pr[\settled(v^{(t+1)L}) - \settled(v^{tL}) = 1|v^{tL} = v]\geq 0.001$, we have $\bb E[\settled(v^{(t+1)L}) - \settled(v^{tL})|v^{tL} = v]\geq 0.001$, and by \autoref{lm:submartingale} we have $\bb E[\otllr(v^{(t+1)L}) - \otllr(v^{tL})|v^{tL} = v]\geq 0$. Summing these two inequalities up proves the lemma.

From now on, we consider the harder case where $\Pr[\settled(v^{(t+1)L}) - \settled(v^{tL}) = 1|v^{tL} = v]< 0.001$. We first prove that
\begin{equation}
\label{eq:progress}
\Pr[\otllr(v^{(t+1)L}) - \otllr(v^{tL})\geq 3|v^{tL} = v]\geq 0.8.
\end{equation}
Recall that in this phase $T_\bootstrap$ runs the $(0.1,25)$-overall likelihood booster $A^{(v_*)}$ for $(\gD_0|_{v_*})^k$ and $(\gD_1|_{v_*})^k$ on the samples $j_1, \dots, j_k$. If $\settled(v^{(t+1)L}) - \settled(v^{tL}) = 0$, i.e. no sample becomes settled in this phase, then
\begin{equation*}
    \otllr(v^{(t+1)L}) - \otllr(v^{tL}) = \sum_{s = 1}^k\left(\log\frac{\gD_1(v^{(t+1)L}_{j_s})}{\gD_0(v^{(t+1)L}_{j_s})} - \log\frac{\gD_1(v^{tL}_{j_s})}{\gD_0(v^{tL}_{j_s})}\right).
\end{equation*}
Conditioning on $v^{tL} = v$, we have $v^{tL}_{j_s}= v_*$, since $v_{j_1} = \cdots = v_{j_k} = v_*$. From $\frac{\gD_b(v^{(t+1)L}_{j_s})}{\gD_b(v_*)} = \gD_b|_{v_*}(v^{(t+1)L}_{j_s})$, we see that 
\begin{equation*}
    \otllr(v^{(t+1)L}) - \otllr(v^{tL}) = \log\prod_{s = 1}^k\frac{\gD_1|_{v_*}(v^{(t+1)L}_{j_s})}{\gD_0|_{v_*}(v^{(t+1)L}_{j_s})}.
\end{equation*}
Therefore, in order to prove (\ref{eq:progress}) by a union bound, we only need to prove that the extended version of phase $t+1$ where $A^{(v_*)}$ gets to run without early halting achieves $\prod_{s = 1}^k\frac{\gD_1|_{v_*}(v^{(t+1)L}_{j_s})}{\gD_0|_{v_*}(v^{(t+1)L}_{j_s})}\geq e^3$ with probability at least 0.9. This is indeed true because $A^{(v_*)}$ is a ($0.1,25$)-overall likelihood booster for $(\gD_0|_{v_*})^k$ and $(\gD_1|_{v_*})^k$.

We now prove $\bb E[\otllr(v^{(t+1)L}) - \otllr(v^{tL})|v^{tL} = v]\geq 0.001$. We prove it by contradiction. Suppose $\bb E[\otllr(v^{(t+1)L}) - \otllr(v^{tL})|v^{tL} = v]<0.001$. 
For $tL\leq s < (t+1)L$, define $\Delta(v^s)$ as the conditional expectation $\bb E[\otllr(v^{s+1}) - \otllr(v^s)|v^s]$ and $\Delta_2(v^s)$ as the conditional variance $\bb E[(\otllr(v^{s+1}) - \otllr(v^s) - \Delta(v^s))^2|v^s]$. Note that 
\begin{equation*}
\Delta_2(v^s) = \bb E[((\otllr(v^{s+1}) - \otllr(v^s))^2|v^s] - (\Delta(v^s))^2\leq \bb E[((\otllr(v^{s+1}) - \otllr(v^s))^2|v^s].
\end{equation*}
Thus by \autoref{lm:submartingale}, we know that $\Delta(v^s)\geq 0.001 \cdot \Delta_2(v^s)\geq 0$.
Now we have
\begin{align*}
0.001> & \bb E[\otllr(v^{(t+1)L}) - \otllr(v^{tL})|v^{tL} = v]\\
= & \sum_{tL\leq s < (t+1)L}\bb E[\Delta(v^s)|v^{tL} = v].
\end{align*}

By Markov's inequality, we have $\Pr\left[\sum_{tL\leq s < (t+1)L}\Delta(v^s)\geq 1|v^{tL} = v\right]\leq 0.001$. Now by a union bound with (\ref{eq:progress}), we have
\begin{align}
  & \bb E\left[\left.\left(\sum_{tL\leq s < (t+1)L}(\otllr(v^{s+1}) - \otllr(v^s) - \Delta(v^s))\right)^2\right|v^{tL} = v\right]\nonumber \\
  = &
  \bb E\left[\left.\left((\otllr(v^{(t+1)L}) - \otllr(v^{tL})) - \sum_{tL\leq s < (t+1)L}\Delta(v^s)\right)^2\right|v^{tL} = v\right]\nonumber \\
  \geq & (0.8 - 0.001) \times (3 - 1)^2\nonumber \\
  > & 3.\label{eq:progress2}
\end{align}
Since $\bb E\left[\otllr(v^{s+1}) - \otllr(v^{s}) - \Delta(v^{s})|v^s\right] = 0$, we have
\begin{equation*}
\bb E\left[(\otllr(v^{s_1+1}) - \otllr(v^{s_1}) - \Delta(v^{s_1}))\cdot (\otllr(v^{s_2+1}) - \otllr(v^{s_2}) - \Delta(v^{s_2}))|v^{tL} = v\right] = 0
\end{equation*}
whenever $s_1< s_2$ by further conditioning on $v^{s_2}$. Thus expanding (\ref{eq:progress2}) we have
\begin{align*}
  & \sum_{tL\leq s < (t+1)L}\bb E[\Delta_2(v^s)|v^{tL} = v]\\
  = &\bb E\left[\left.\sum_{tL\leq s < (t+1)L}(\otllr(v^{s+1}) - \otllr(v^s) - \Delta(v^s))^2\right|v^{tL} = v\right]\geq 3.
\end{align*}
Since $\Delta(v^s)\geq 0.001\cdot\Delta_2(v^s)$, we have
\begin{equation*}
  0.001 > \sum_{tL\leq s < (t+1)L}\bb E[\Delta(v^s)|v^{tL} =  v] \geq 0.001\cdot \sum_{tL\leq s < (t+1)L}\bb E[\Delta_2(v^s)|v^{tL} = v] \geq 0.001\times 3,
\end{equation*}
a contradiction.

Now we have shown $\bb E[\otllr(v^{(t+1)L}) - \otllr(v^{tL})|v^{tL} = v]\geq 0.001$. Adding it to the trivial inequality $\bb E[\settled(v^{(t+1)L}) - \settled(v^{tL})|v^{tL} = v]\geq 0$ proves the lemma.
\end{proof}

\subsection{A helper inequality} \label{sec:helper}

\begin{lemma}
\label{lm:helper}
For all $M\geq 0, t\in (0,e^{M}]$, we have $$t\ln t - (t - 1)\geq \frac 1{M + 2}\cdot t\ln^2 t.$$
\end{lemma}
\begin{proof}
  Define function $h(t) = t\ln t - (t - 1) - \frac 1{M + 2}\cdot t\ln^2 t$ on the interval $t\in (0,e^{M}]$. Our goal is to show $h(t)\geq 0$. Note that $h(1) = 0$, so we only need to show $h'(t)\geq 0$ for $t\geq 1$ and $h'(t)\leq 0$ for $t\leq 1$. We prove this by calculating $h'(t)$:
  \begin{equation*}
  h'(t) = \ln t - \frac 1{M + 2} \cdot \ln^2 t - \frac 2{M + 2}\cdot \ln t = \left(1 - \frac{(\ln t) + 2}{M + 2}\right)\ln t.
  \end{equation*}
  Note that $1-\frac{(\ln t) + 2}{M + 2}\geq 0$ because $\ln t\leq M$. Therefore $h'(t)\geq 0$ when $t\geq 1$ and $h'(t)\leq 0$ when $t\leq 1$, as desired.
\end{proof}

\section{Application 1: Selection Problem} \label{sec:selection}

\subsection{Bi-correlated samples} \label{sec:bicorr}

To establish a relationship between correlated samples and selection, we first define an intermediate problem. The \emph{bi-correlated samples problem} is defined by (here $\gD_{ab}\coloneqq\gD_a\times\gD_b$):
\begin{align*}
\bicorr_\epsilon(f,\gD)~
&\textstyle\coloneqq~ \min_{k\geq 1}\,
\textstyle\dt_\epsilon(f^{2k}, \frac{1}{2}\gD_{01}^k + \frac{1}{2}\gD_{10}^k),\\
\bicorr_\epsilon(f)~
&\textstyle\coloneqq~ \max_{\gD}\,
\textstyle\bicorr_\epsilon(f,\gD).
\end{align*}
That is, the task is to decide whether $f^{2k}$ outputs $(01)^k$ or $(10)^k$ as $k\to\infty$. We show this is as hard as correlated samples:
\begin{lemma} \label{lem:bicorr}
$\corr(f,\gD) = \Theta(\bicorr(f,\gD))$.
\end{lemma}
\begin{proof}
It is obvious that $\bicorr(f,\gD)\leq \corr(f,\gD)$, so we focus on the converse, $\corr(f,\gD) \leq O(\bicorr(f,\gD))$. The proof is via a hybrid argument. Let $T\colon(\{0,1\}^n)^{2k}\to\{0,1\}$ be an optimal algorithm for $\bicorr_{1/3}(f,\gD)$ that uses $k$ sample pairs. Letting $d(-,-)$ denote the statistical distance between two distributions, the fact that $T$ achieves error $\epsilon\coloneqq 1/3$ can be written as
\[
d(T(\gD_{01}^k),T(\gD_{10}^k))~\geq~ 1-2\epsilon.
\]
By the triangle inequality,
\[
d(T(\gD_{01}^k),T(\gD_{00}^k))\, +\,
d(T(\gD_{00}^k),T(\gD_{10}^k))~\geq~ 1-2\epsilon.
\]
Either the first or the second term is $\geq (1-2\epsilon)/2$. Say the first (second case is similar):
\[
d(T(\gD_{01}^k),T(\gD_{00}^k))~\geq~ (1-2\epsilon)/2~=~ 1 - 2\epsilon'
\qquad\text{where}\enspace
\epsilon' \coloneqq  1/4+\epsilon/2 = 5/12.
\]
This means we can turn $T$ into an $5/12$-error algorithm for the correlated $k$-samples problem: the odd numbered input samples of $T$ the algorithm can generate from $\gD_0$ on its own; the even numbered input samples of $T$ are taken from the input to the correlated $k$-samples problem. Finally, the error can be reduced to $1/3$ via \autoref{fact:amplify}.
\end{proof}

\subsection{Proof of \autoref{thm:selection}}
\paragraph{First item.}
The following claim together with \autoref{lem:bicorr} implies the first item.
\begin{claim}
$\bicorr_\epsilon(f,\gD)\leq\sel_\epsilon(f,\gD)$.
\end{claim}
\begin{proof}
Let $T_\sel$ be an optimal algorithm for $\sel_\epsilon(f,\gD)$ using $k$ samples. We describe an algorithm $T_\bicorr$ for bi-correlated $k$-samples with the same error and query cost. Let $x=(x_{ij})$ for $(i,j)\in[k]\times[2]$ be the random input to $T_\bicorr$, that is, either \emph{(i)} $x\sim \gD_{01}^k$ or \emph{(ii)} $x\sim \gD_{10}^k$. The algorithm $T_\bicorr$ chooses a random string $z\in[2]^k$ and runs $T_\sel$ on input $y\coloneqq (x_{iz_i})_{i\in[k]}$. Note that $y$ is distributed as $\gD^k$ in both cases \emph{(i)} and \emph{(ii)}. Suppose $T_\sel$ outputs some $(i,f(x_{iz_i}))$. Assuming this output is correct for selection, and remembering our choice of $z_i$, we can deduce which case, \emph{(i)} or \emph{(ii)}, the input $x$ came from, and let $T_\bicorr$ guess accordingly. Hence algorithm $T_\bicorr$ is correct every time $T_\sel$ is, and so the error parameter is unaffected.
\end{proof}

\paragraph{Second and third item.}
For separating correlated samples from selection, we again consider the $n$-bit $\xor_n$ function. Define $x\sim \gD$ by the following process:
\begin{enumerate} \setlength{\itemsep}{-0.2em}
	\item Sample $z$ uniformly from $\{0,1\}^{n-2}$ and let $a\coloneqq\xor_{n-2}(z)$.
	\item Sample $b$ uniformly from $\{0,1\}$.
	\item With probability $\epsilon\coloneqq 1\%$, output $x\coloneqq aaz$; with probability $1-\epsilon=99\%$, output $x\coloneqq bbz$.
\end{enumerate}

Note that the first two bits of $x\sim\gD$ are identical and hence $\xor_n(x)=\xor_{n-2}(z)$. Moreover, the first bit is $\epsilon$-correlated with the function value $\xor_n(x)$. This makes $(\xor_n,\gD)$ easy for the correlated samples problem: The $1$-query algorithm that guesses the function value based on the first bit of the first sample has error $\leq 1/2-\epsilon/2$, and this error can be reduced to $1/3$ via \autoref{fact:amplify}. This shows that $\corr(\xor_n,\gD)=O(1)$.

Next we prove the lower bound $\sel(\xor_n,\gD)=\Omega(n)$, which also proves the third item. Suppose for contradiction that $T$ is a height-$(n-3)$ deterministic decision tree for $k$-selection for $(\xor_n,\gD)$. Consider any leaf $\ell$ that claims the $i$-th sample evaluates to $b\in\{0,1\}$. If we condition $\gD^k$ by the $\leq n-3$ queries made by $\ell$, we note that the function value is still only slightly biased away from $1/2$, that is, $\E_{x\sim \gD^k|\ell}[\xor_n(x_i)] \in 1/2\pm \epsilon$. Hence no leaf of $T$ can compute selection to within error $\leq 1/3$. This concludes the proof of \autoref{thm:selection}.

\section{Application 2: Randomized Composition} \label{sec:composition}

\paragraph{Goal.}
In this section we prove \autoref{thm:composition}, namely $\rdt(f\circ g)\geq\Omega(\fbs(f)\rdt(g))$. By \autoref{thm:main} and \autoref{lem:bicorr} (from \autoref{sec:bicorr}) it suffices to show
\[
\bicorr(g)~\leq~O(\rdt(f\circ g)/\fbs(f)).
\]
Let $T$ be an optimal $1/10$-error algorithm for $f\circ g$ making $q\coloneqq O(\rdt(f\circ g))$ queries. Our goal is, given any balanced input distribution $\gD\coloneqq \frac{1}{2}\gD_0 + \frac{1}{2}\gD_1$ to the inner function $g$, to build a bounded-error algorithm $T'$ solving the bi-correlated samples problem for $(g,\gD)$.

\paragraph{Rarely queried block.}
By the definition of $\fbs(f)$, there is an input $y\in\{0,1\}^n$ to $f$ (say, $f(y) = 0$) with sensitive blocks $B_1,\cdots,B_N\subseteq [n]$ and weights $w_1,\cdots,w_N\in[0,1]$ such that
\begin{align}
    \textstyle\sum_{j\in[N]} w_j~&=~ \fbs(f),\label{eq:41}\\
    \textstyle\sum_{j:B_j\ni i}w_j~&\leq~ 1,\qquad\forall i\in[n].\label{eq:42}
\end{align}
For any $z\in\{0,1\}^n$, define $\gD_z$ as the distribution over $(x_1,\cdots,x_n)\in (\{0,1\}^m)^n$ where each $x_i$ is drawn independently from $\gD_{z_i}$. Hence we have $g^n(x)=z$ for $x\sim\gD_z$. We define
\[
q_j~\coloneqq~\text{expected \# of queries $T$ makes to block $B_j$ on input $\gD_y$}.
\]
That is, if we denote by $i_t\in[n]$ the block that $T$ queries at time $t$, then $q_j$ is the expected number of time steps $t$ with $i_t\in B_j$. By linearity of expectation and (\ref{eq:42}), we have
\begin{equation*}\textstyle
\sum_{j\in[N]} w_jq_j
~=~ {\bb E}\Big[\sum_{j\in[N]} w_j \sum_{t:i_t\in B_j}1\Big]
~=~ {\bb E}\Big[\sum_{t \in[q]} \sum_{j:B_j\ni i_t}w_j\Big]
~\leq~ {\bb E}\Big[ \sum_{t \in[q]}1\Big]
~\leq~q.
\end{equation*}
Combining this with (\ref{eq:41}), we know there exists $j\in[N]$, say $j=1$ for simplicity, such that
\[
q_1~\leq~\frac{q}{\fbs(f)}.
\]

\paragraph{Truncated $T$.}
Next we modify $T$ so that it makes at most $5q_1$ queries to block $B_1$ for \emph{every} input (not just on average over $\gD_y$). Namely, if $T$ makes more than $5q_1$ queries to block $B_1$, we simply let $T$ halt and output $1$; otherwise its behavior is unchanged. We denote this ``truncated'' algorithm by $T^\tr$. We claim that~$T^\tr$ still computes $f\circ g$ correctly on average over both $\gD_{y}$ and $\gD_{y^{B_1}}$ (recall that $y^{B_1}$ is $y$ but with the block $B_1$ flipped; note that $f(y^{B_1})=1$ and hence $(f\circ g)(x)=1$ for each $x \sim \gD_{y^{B_1}}$)
\begin{align}
\textsl{Correct for $x\sim \gD_{y^{B_1}}$:}\qquad
\Pr[T^\tr(x)=1]
&~\geq~  \Pr[T(x)=1]\nonumber \\
&~\geq~ 4/5 \label{eq:54}. \\[3mm]
\textsl{Correct for $x\sim \gD_y$:}\qquad 
\Pr[T^\tr(x)=0]
&~\geq~ \Pr[T(x)=0] - \Pr[T(x)\textup{ makes}>5q_1\textup{ queries to }B_1]\label{eq:ub} \\
&~\geq~ 4/5 - 1/5\label{eq:mb}\\
&~=~ 3/5,\label{eq:53}
\end{align}
where (\ref{eq:ub}) uses the Union Bound and (\ref{eq:mb}) uses the Markov Bound.

\paragraph{Algorithm $T'$.}
We are ready to define the algorithm $T'$ for the bi-correlated samples problem for $(g,\gD)$. The random input to this problem is $z=(z_{ij})$, $(i,j)\in[n]\times\{0,1\}$, sampled either from \emph{(i)} $\gD_{01}^n$ or \emph{(ii)} $\gD_{10}^n$. On input $z$ the algorithm $T'$ simply runs $T^\tr$ on the input $(x_1,\ldots,x_n)\in(\{0,1\}^m)^n$ defined by
\[
x_i~\coloneqq~
\begin{cases}
z_{iy_i} & \text{for $i\in B_1$},\\
\sim\gD_{y_i} & \text{for $i\notin B_1$}.
\end{cases}
\]
That is, for $i\in B_1$ the algorithm $T'$ simply copies its input bits in $z$ to the bits of $x$. For $i\notin B_1$ the algorithm~$T'$ uses its own randomness to generate an independent sample from either $\gD_0$ or $\gD_1$. The key observation is that in case \emph{(i)} we have $x\sim \gD_{y}$, and in case \emph{(ii)} we have $x\sim\gD_{y^{B_1}}$. But $T^\tr$ can distinguish these two cases to within bounded error by (\ref{eq:54}) and (\ref{eq:53}). Hence $T'$ is a bounded-error algorithm for bi-correlated samples with query cost $5q_1\leq O(q/\fbs(f))$. This completes the proof of \autoref{thm:composition}.
\subsection*{Acknowledgments}

We thank Shalev Ben-David for correspondence about their ongoing work~\cite{Ben-David2020}. AD thanks Mark Braverman for interesting discussions of related topics. LH is supported in part by NSF Award IIS-1908774. WM is supported by a Stanford Graduate Fellowship.  LYT is supported by NSF grant CCF-1921795 and CAREER Award CCF-1942123.

\addcontentsline{toc}{section}{References}
\DeclareUrlCommand{\Doi}{\urlstyle{sf}}
\renewcommand{\path}[1]{\small\Doi{#1}}
\renewcommand{\url}[1]{\href{#1}{\small\Doi{#1}}}
\bibliographystyle{alphaurl}
\bibliography{ref}

\end{document}